\ifdefined\Twocolumn
\documentclass[journal]{IEEEtran}
\else
\documentclass[12pt, draftclsnofoot, onecolumn]{IEEEtran}
\fi

\usepackage{amsmath,graphicx,mymathdefs}
\usepackage{color}
\usepackage{enumerate}
\usepackage[none]{hyphenat}
\usepackage{etoolbox} 

\newtheorem{corollary}{Corollary}
\makeatletter
\newcounter{remark}

\def\remark{\addtocounter{remark}{1}\def\@currentlabel{\theremark}%
\emph{Remark~\theremark}. } \makeatother

\newtoggle{twocolumn}
\ifCLASSOPTIONtwocolumn\toggletrue{twocolumn}\else\togglefalse{twocolumn}\fi

\iftoggle{twocolumn}{
  \def\figwidth{.95\linewidth}}{
  \def\figwidth{.58\linewidth}}

\def\C#1x#2{\setC^{#1\times #2}}
\def\CN{${\mathcal C\mathcal N}(0,1)$}

\def\Pd{\rho_d}
\def\hH{\widehat{\mathbf{H}}}
\def\tH{\widetilde{\mathbf{H}}}
\def\Yp{\mathbf{Y}_\text{p}}
\def\figref#1{Fig.~\ref{#1}}
\def\thref#1{Theorem.~\ref{#1}}
\def\remref#1{Remark~\ref{#1}}

\def\bx{\mathbf x}

\def\bG{\mathbf G}
\def\bH{\mathbf H}
\def\bI{\mathbf I}

\def\adag{\alpha^\dag}
\def\rhomax{\rho_\text{max}}
\def\maximize{\mathop{\text{maximize}}}
\def\leref#1{Lemma~\ref{#1}}
\def\figref#1{Fig.~\ref{#1}}
\newtheorem{lemma}{Lemma}
\newtheorem{theorem}{Theorem}

\iftoggle{twocolumn}{}{\linespread{1.45}\normalsize}

\title{Achievable Throughput and Training Optimization of Uplink
Multiuser Massive MIMO Systems}

\author{Songtao Lu and Zhengdao Wang

\thanks{Parts of this work were presented at 2014 IEEE Global Communications Conference (GLOBECOM), Austin, TX USA, Dec. 8--12, 2014 \cite{luwa14g3} and 2015 IEEE Wireless Communications and Networking Conference (WCNC), New Orleans, LA USA, Mar. 9--12, 2015 \cite{luwa15}.}
\thanks{Songtao Lu and Zhengdao Wang are
with the Department of Electrical and Computer Engineering,
    Iowa State University, Ames, IA 50011, USA (e-mail: songtao@iastate.edu;
zhengdao@iastate.edu).} }

\begin{document}
\maketitle

\begin{abstract}
We study the performance of uplink transmission in a large-scale (massive)
MIMO system, where all the transmitters have single antennas and the base
station has a large number of antennas. Specifically, we first derive the
rates that are possible through minimum mean-squared error (MMSE) channel
estimation and three linear receivers: maximum ratio combining (MRC),
zero-forcing (ZF), and MMSE. Based on the derived rates, we quantify the
amount of energy savings that are possible through increased number of
base-station antennas or increased coherence interval. We also analyze
achievable total degrees of freedom (DoF) of such a system without assuming
channel state information at the receiver, which is shown to be the same as
that of a point-to-point MIMO channel. Linear receiver is sufficient to
achieve total DoF when the number of users is less than the number of
antennas. When the number of users is equal to or larger than the number of
antennas, nonlinear processing is necessary to achieve the full degrees of
freedom. Finally, the training period and optimal training energy allocation
under the average and peak power constraints are optimized jointly to maximize
the achievable sum rate when either MRC or ZF receiver is adopted at the
receiver.
\end{abstract}


\iftoggle{twocolumn}{}{\clearpage}

\section{Introduction}

Massive multiple-input multiple-output (MIMO) systems are a type of cellular
communication where the base station is equipped with a large number of
antennas. The base station serves multiple mobile stations that are usually
equipped with a small number of antennas, typically one. Massive MIMO holds
good potential for improving future communication system performance. There
are several challenges with designing such massive MIMO systems, including
e.g., channel state information (CSI) acquisition \cite{sozo14}, base station
received signal processing \cite{gkka14}, downlink precoding with imperfect
CSI \cite{mola13}, signal detection algorithm \cite{haha14}, etc. For
multi-cell system, pilot contamination and inter-cell interference also need
to be dealt with \cite{joas11}. There is already a body of results in the
literature about the analysis and design of large MIMO systems; see e.g., the
overview articles \cite{rpll13,laed14} and references there in.

To reveal the potential that is possible with massive MIMO systems, it is
important to quantify the achievable performance of such systems in realistic
scenarios. For example, it is too optimistic to assume that perfect CSI can be
acquired at the base station in the uplink, because such acquisition takes
time, energy, and channel estimation error will always exist. For the
downlink, in order to perform effective beamforming, CSI is again needed,
which needs to be either estimated by the mobile stations and then fed back to
the base station, which is a non-trivial task, or, acquired by the base
station by exploiting channel reciprocity in a time-division duplexing setup.

\subsection{Scope of this paper}

In this paper, we are interested in performance of the \emph{uplink}
transmission in a \emph{single-cell} system. In particular, we ask what rates
can be achieved in the uplink by the mobile users if we assume realistic
channel estimation at the base station. Similar analysis has been performed in
\cite{nglm13,homa13,qish14}, but the analysis therein assumes equal power
transmission during the channel training phase and the data transmission
phase. Also, the effect of channel coherence interval on system throughput was
discussed in \cite{emmi14} and power allocation and training duration
optimization for the uplink MIMO system were considered in \cite{luwa14g3} for
single-cell system and in \cite{ngma14} for a multi-cell system. However, peak
power constraint was not considered. For a fixed training period, to obtain an
accurate estimate the training power needs to be high to enable enough
training energy. As a result, peak power constraint, if present, may be
violated. The solution is to optimize the training duration also.

If we allow the users to cooperate, then the system can be viewed as a
point-to-point MIMO channel. The rates obtained in \cite{haho03}, and the
stronger result on non-coherent MIMO channel capacity in \cite{zhts02} can
serve as upper bound for the system sum rate. The question is how much of this
sum rate can be achieved without user cooperation and without using elaborate
signaling such as signal packing on Grassmannian manifolds.

For a system with $K$ mobile users, $M$ base station antennas, and block
fading channel with coherence interval $T$, we derive achievable rate using
linear channel estimation and linear base station (front-end) processing,
including maximum ratio combining (MRC), zero-forcing (ZF), and minimum
mean-squared estimation (MMSE) processing. The total degrees of freedom (DoF)
is also quantified. We also quantify the needed transmission power for
achieving a given rate, when $M\gg 1$, which is an refinement of the
corresponding result in \cite{nglm13}.

Furthermore, the energy allocation and training duration are also both
optimized for uplink multi-user (MU) MIMO systems in a systematic way. Two
linear receivers, MRC and ZF, are adopted with imperfect CSI. The average and
peak power constraints are both incorporated. We analyze the convexity of this
optimization problem, and derive the optimal solution. The solution is in
closed form except in one case where a one-dimensional search of a
quasi-concave function is needed. Simulation results are also provided to
demonstrate the benefit of optimized training, compared to equal power
allocation considered in the literature.

The main contributions of this paper are listed as follows:
\begin{enumerate}

\item We consider the energy allocation between training phase and data phase,
and derive achievable rate using linear channel estimation and linear base
station (front-end) processing, including MRC, ZF and MMSE processing.

\item We quantify the total degrees of freedom (DoF) with estimated channels.

\item We quantify the needed transmission power for achieving a given rate,
when $M\gg 1$, which is an refinement of the corresponding result in
\cite{nglm13}.

\item We provide a complete solution for the optimal training duration and
training energy in an uplink MU-MIMO system with both MRC and ZF receiver,
under both peak and average power constraints.
\end{enumerate}

\subsection{Related Works} The throughput of massive MIMO systems has been
studied in several recent papers. For example, the issue of non-ideal hardware
and its effect on the achievable rates were investigated in
\cite{bjho14,emmi14}. In \cite{nglm13,nglm13m}, the achievable rates with
perfect or estimated CSI were derived and scaling laws were obtained in terms
of the power savings as the number base station antennas is increased. For
channel estimation, the training power and training duration were not
optimized for rate maximization. Expressions for uplink achievable rates under
perfect or imperfect CSI were derived in \cite{qish14} for Ricean channels
with an arbitrary-rank deterministic component. For the downlink MIMO
broadcast channel, the optimization over training period and power in both
training phase and feedback phase was investigated in~\cite{koji11}. An
optimized energy reduction scheme was proposed in~\cite{miao13} for uplink MU
MIMO in a single cell scenario, where both RF transmission power and circuit
power consumption were incorporated. Wireless energy transfer using massive
MIMO was considered in \cite{yhzy15}, where the uplink channels were estimated
and the uplink rate for the worst user was maximized. Downlink throughput
scaling behavior was investigated in \cite{beps15}, where it was shown that
unused uplink throughput can be used to trade off for downlink throughput and
that the downlink throughput is proportional to the logarithm of the number of
base-station antennas. Stochastic geometry was used \cite{bahe15} to analyze
the uplink SINR and rate performance of large-scale massive MIMO systems with
maximum ratio combining or zero-forcing receivers. Scaling laws between the
required number of antennas and the number of users to maintain the same
signal to interference ratio (SIR) distributions were derived.

The rest of the paper is organized as follows. A system model is developed in
\secref{systemmodel}. The channel estimation is discussed in
\secref{channelesti}. The achievable rates for linear receivers and the
achievable total degrees of freedom of the system are derived in
\secref{achirate}. \secref{durapeak} formulates the problem of maximization of
achievable rate with both average and peak power constraints. The solution of
this optimization problem is presented in Sections~\ref{sec.sola} and
\ref{sec.solb}. When $M\to \infty$, simplified expressions for the achievable
rates are discussed in \secref{massivecase}. Numerical simulation results are
reported in \secref{sec:simulations} and finally conclusions are drawn in
\secref{sec:conclusion}.

\section{System Model}\label{systemmodel}

Notation: We use $\mathbf{A}^\dag$ to denote the Hermitian transpose of a
matrix $\mathbf{A}$, $\mathbf{I}_K$ to denote a \size KxK identity matrix,
$\setC$ to denote the complex number set, $\lfloor\cdot \rfloor$ to denote the
integer floor operation, \iid to denote ``independent and identically
distributed'', and $\mathcal{CN}(0,1)$ to denote circularly symmetric complex
Gaussian distribution with zero mean and unit variance.

Consider a single-cell uplink system, where there are $K$ mobile users and one
base station. Each user is equipped with one transmit antenna, and the base
station is equipped with $M$ receive antennas. The received signal at the base
station is expressible as
\begin{equation}\label{eq.sys}
 \mathbf{y} = \mathbf{H}\mathbf{s} +\mathbf{n}
\end{equation}
where $\mathbf{H}\in \C MxK$ is the channel matrix, $\mathbf{s}\in \C Kx1$ is
the transmitted signals from all the $K$ users; $\mathbf{n}\in \C Mx1$ is the
additive noise, $\mathbf{y}\in \C Mx1$ is the received signal. We make the
following assumptions:

A1) The channel is block fading such that within a \emph{coherence interval}
of $T$ channel uses, the channel remains constant. The entries of $\mathbf{H}$
are \iid and taken from \CN. The channel changes independently from block to
block. The CSI is neither available at the transmitters nor at the receiver.

A2) Entries of the noise vector $\mathbf{n}$ are \iid and from \CN. Noises in
different channel uses are independent.

A3) The average transmit power per user per symbol is $\rho$. So within a
coherence interval the total transmitted energy is $\rho T$.

In summary, the system has four parameters, $(M,K,T,\rho)$. We will allow the
system to operate in the ergodic regime, so coding and decoding can occur over
multiple coherent intervals.

\section{Channel Estimation} \label{channelesti}

We assume that $K\le M$ and $K<T$ in this section. To derive the achievable
rates for the users, we use a well-known scheme that consists of two phases
(see e.g., \cite{haho03}):

\noindent \emph{Training Phase}. This phase consists of $T_{\tau}$ time
intervals. The $K$ users send time-orthogonal signals at power level
$\rho_{\tau}$ per user. The training signal transmitted can be represented by
a $K\times T_{\tau}$ matrix $\mathbf{\Phi}$ such that
$\mathbf{\Phi}\mathbf{\Phi}^\dag=E\mathbf{I}_K$, where $E=\rho_{\tau}T_{\tau}$
is the total training energy per user per coherent interval. Note that we
require $T_{\tau}\ge K$ to satisfy the time-orthogonality.

\noindent \emph{Data Transmission Phase}. Information-bearing symbols are
transmitted by the users in the remaining $T_d=T-T_{\tau}$ time intervals. The
average power per symbol per user is $\Pd=(\rho T-E)/T_d$.

\subsection{MMSE Channel Estimation}\label{mmsechannelesti}

In the training phase, we will choose $\mathbf{\Phi}=\sqrt{E}\mathbf{I}_K$ for
simplicity. Other scaled unitary matrix can also be used without affecting the
achievable rate. Note that the transmission power is allowed to vary from the
training phase to the data transmission phase. With our choice of
$\mathbf{\Phi}$, the received signal $\mathbf{Y}_p \in \mathbb{C}^{M\times
T_{\tau}}$ during the training phase can be written as
\begin{equation}
  \mathbf{Y}_p=\mathbf{H}\mathbf{\Phi} + \mathbf{N} =
  \sqrt{E}\mathbf{H}+\mathbf{N}
\end{equation}
where $\mathbf{N}\in \mathbb{C}^{M\times T_{\tau}}$ is the additive noise. The
equation describes $M\times T_{\tau}$ independent identities, one for each
channel coefficient. The (linear) MMSE estimate for the channel $\mathbf{H}$
is given by
\begin{equation}
  \hH = \frac{\sqrt E}{E+1}  \Yp = \frac{E}{E+1}\mathbf{H} + \frac{\sqrt
  E}{E+1}\mathbf{ N}.
\end{equation}
The channel estimation error is defined as
\begin{equation}
\tH=\mathbf{H}-\hH=\frac 1{E+1}\mathbf{H} - \frac{\sqrt E}{E+1}\mathbf{N}.
\end{equation}
It is well known and easy to verify that the elements of $\hH$ are \iid
complex Gaussian with zero mean and variance
\def\sigmah{\sigma_{\hH}^2}
\begin{equation}
  \sigmah = \frac{E}{E+1},
\end{equation}
and the elements of $\tH$ are \iid complex Gaussian with zero mean and
variance
\def\sigmat{\sigma_{\tH}^2}
\begin{equation}
  \sigmat = \frac{1}{E+1}.
\end{equation}
Moreover, $\hH$ and $\tH$ are in general uncorrelated as a property of linear
MMSE estimator, and in this case independent thanks to the Gaussian
assumptions.

\subsection{Equivalent Channel} \label{equichan}

Once the channel is estimated, the base station has $\hH$ and will decode the
users' information using $\hH$. We can write the received signal as
\begin{equation} \label{eq.equiv}
 \mathbf{ y}=\hH \mathbf{s} + \tH \mathbf{s} + \mathbf{n}  := \hH \mathbf{s} +
 \mathbf{v}
\end{equation}
where $\mathbf{v}\bydef \tH \mathbf{s} + \mathbf{n}$ is the new equivalent
noise containing actual noise $\mathbf{n}$ and self interference $\tH\mathbf{
s}$ caused by inaccurate channel estimation. Assuming that each element of
$\mathbf{s}$ has variance $\Pd$ during the data transmission phase, and there
is no cooperation among the users, the variance of each component of
$\mathbf{v}$ is
\def\sigmav{\sigma_v^2}
\begin{equation}\label{eq.varv}
\sigmav=\frac{K\Pd}{E+1}+1.
\end{equation}
If we replace $\mathbf{v}$ with a zero-mean complex Gaussian noise with equal
variance $\sigmav$, but independent of $\mathbf{s}$, then the system described
in \eqref{eq.equiv} can be viewed as MIMO system with perfect CSI at the
receiver, and equivalent signal to noise ratio (SNR)
\def\rhoe{\rho_{\text{eff}}}
\begin{equation}\label{eq.eff}
 \rhoe \bydef \frac {\Pd\sigmah}{\sigmav} = \frac{\Pd E}{K\Pd + E +1}
 =\frac{\Pd}{1+\frac{K\Pd+1}E}.
\end{equation}
The SNR is the signal power from a single transmitter per receive antenna
divided by the noise variance per receive antenna. It is a standard argument
that a noise equivalent to $\mathbf{v}$ but assumed independent of
$\mathbf{s}$ is ``worse'' (see e.g., \cite{haho03}). As a result, the derived
rate based on such assumption is achievable. In the following, for notational
brevity, we assume that $\mathbf{v}$ in \eqref{eq.equiv} is independent of
$\mathbf{s}$ without introducing a new symbol to represent the equivalent
\emph{independent} noise.

Note that the effective SNR $\rhoe$ is the actual SNR $\Pd$ divided by a loss
factor $1+(K\Pd+1)/E$. The loss factor can be made small if the energy $E$
used in the training phase is large.

\subsection{Energy Splitting Optimization} \label{sec.split}

The energy in the training phase can be optimized to maximize the effective
SNR $\rhoe$ in \eqref{eq.eff} for point-to-point MIMO system, as has been done
in \cite[Theorem~2]{haho03}. We adapt only the result below for our case
because it is relevant to our discussion. Importantly, with the effective SNR
adopted in this paper, the achievable rate with MRC, ZF and MMSE receiver can
be easily optimized in a closed form.

We assume the average transmitted power over one coherence interval $T$ is
equal to a given constant $\rho$, namely
\(
\rho_dT_d+\rho_{\tau}T_{\tau}=\rho T \label{eq.ener}
\).
Let $\alpha:= \rho_\tau T_\tau/(\rho T)$ denote the fraction of the total
transmit energy that is devoted to channel training; i.e.,
\begin{equation} \label{eq.optfra}
\rho_{\tau}T_{\tau}=\alpha\rho T, \quad \rho_dT_d=(1-\alpha)\rho T,
\quad 0\le \alpha\le 1.
\end{equation}

Define an auxiliary variable when $T_d\ne K$:
\begin{equation}
  \gamma\bydef \frac{K\rho T+T_d}{\rho T(T_d-K)}
\end{equation}
which is positive if $T_d>K$ and negative if $T_d<K$.

It can be easily verified that in all the three cases, namely $T_d=K$,
$T_d>K$, and $T_d<K$, $\rhoe$ is concave in $\alpha$ within $\alpha\in(0,1)$.
The optimal value for $\alpha$ that maximizes $\rhoe$ is given as follows:
\begin{equation}\label{eq.optalpha}
  \alpha^*=\begin{cases}
  -\gamma+\sqrt{\gamma(\gamma+1)}, & T_d>K \\
  \frac 12, & T_d=K \\
  -\gamma-\sqrt{\gamma(\gamma+1)}, & T_d<K
  \end{cases}
\end{equation}
The maximized effective SNR $\rhoe^*$ is given as
\begin{equation}
  \rhoe^*=\begin{cases}
\frac{\rho T}{T_d-K}(-2\sqrt{\gamma(\gamma+1)}+(1+2\gamma)), & T_d>K \\
\frac{(\rho T)^2}{4K(1+\rho T)}, & T_d=K \\
\frac{\rho T}{T_d-K}(2\sqrt{\gamma(\gamma+1)}+(1+2\gamma)), & T_d<K
\end{cases}
\end{equation}
At high SNR ($\rho\gg 1$), we have
\begin{equation}
\gamma\approx\frac{K}{T_d-K},
\end{equation}
and the optimal values are
\begin{equation}
  \alpha^*_{\textrm{H}}\approx\frac{\sqrt{K}}{\sqrt{T_d}+\sqrt{K}}, \quad
  \rhoe^*\approx\frac{T}{(\sqrt{T_d}+\sqrt{K})^2}\rho.
\end{equation}
At low SNR ($\rho\ll 1$), we have
\begin{equation}
\gamma\approx\frac{T_d}{\rho T(T_d-K)},
\end{equation}
and the optimal values are
\begin{equation}\label{eq.lowsnr}
  \alpha^*_{\textrm{L}}\approx\frac 12, \quad
  \rhoe^*\approx\frac{(\rho T)^2}{4T_d}.
\end{equation}

\section{Achievable rates}\label{achirate}

\subsection{Rates of Linear Receivers}

Given the channel model \eqref{eq.equiv}, linear processing can be applied to
$\mathbf{y}$ to recover $\mathbf{s}$, as in e.g., \cite{nglm13}. Let
$\mathbf{A}\in \C KxM$ denote the linear processing matrix. The processed
signal is
\def\hs{\widehat{\mathbf{s}}}
\begin{equation}
  \hs \bydef \mathbf{A}\mathbf{y} = \mathbf{A}\hH \mathbf{s} + \mathbf{A}\mathbf{v}.
\end{equation}
The MRC processing is obtained by setting $\mathbf{A}=\hH^\dag$. The ZF
processing is obtained by setting $\mathbf{A}=(\hH^\dag\hH)^{-1}\hH^\dag$. And
the MMSE processing is obtained by setting $\mathbf{A}=\hH^\dag (\hH^\dag \hH+
\frac{\sigma_v^2}{\Pd})^{-1}$, where $\sigma_v^2$ is as given in
\eqref{eq.varv}.

Based on the equivalent channel model, viewed as a multi-user MIMO systems
with perfect receiver CSI and equivalent SNR $\rhoe$, the achievable rates
lower bounds derived in \cite[Propositions~2 and~3]{nglm13} can then be
applied. Also, setting the training period equal to the total number of
transmit antennas possesses certain optimality as derived in \cite{haho03},
which means $T^*_{\tau}=K$. Specifically, for MRC the following ergodic sum
rate
 is achievable:
\begin{equation}\label{eq.ratemrc}
  R^\text{(MRC)} \bydef K\left (1 - \frac KT\right) \log_2\left(1 +
  \frac{\rhoe(M-1)}{\rhoe(K-1)+1} \right).
\end{equation}
For ZF, assuming $M>K$, the following sum rate is achievable:
\begin{equation} \label{eq.ratezf}
  R^\text{(ZF)} \bydef K\left (1 - \frac KT\right) \log_2\left(1 +
  \rhoe(M-K) \right).
\end{equation}
Note that the factor $(1-\frac KT)$ is due to the fact that during one
coherence interval of length $T$, $K$ time slots have been used for the
training purpose. The number of data transmission slots is $T-K$, and the
achieved rate needs to be averaged over $T$ channel uses. Also, these rates
are actually lower bounds on achievable rates (due to the usage of Jensen's
inequality).

For the MMSE processing, we assume that the noise $\mathbf{v}$ in
\eqref{eq.equiv} is independent of $\hH$ and Gaussian distributed. The total
received useful signal energy per receive antenna is $K\Pd E/(E+1)$. The noise
variance per receive antenna is $\sigma_v^2$ as in \eqref{eq.varv}. Therefore
the SNR per receive antenna is equal to
\def\snr{\textsf{SNR}}
\begin{equation}
  \snr=\frac{K\Pd E}{K\Pd + E + 1}.
\end{equation}
Using a recent result in \cite[Proposition~1]{mcct10}, we obtain an achievable
sum rate for MMSE processing after MMSE channel estimation, given by
\ifCLASSOPTIONtwocolumn
\begin{multline}
   R^\text{(MMSE)} \bydef K e^{K/\snr}
   \left( f(M, K, \frac K{\snr})-\right.\\
    \left.f(M, K-1, \frac K{\snr})\right) \log_2 e\label{eq.ratemmse}
\end{multline}
\else
\begin{equation}\label{eq.ratemmse}
   R^\text{(MMSE)} \bydef K e^{K/\snr}
   \left( f(M, K, \frac K{\snr})- f(M, K-1, \frac K{\snr})\right) \log_2 e
\end{equation}
\fi where the function $f()$ is defined as
\begin{equation} \label{eq.f}
  f(m,n,x)=\frac{\sum_{k=1}^n \det \mathbf{\Psi}_{n,m}(k,x)}{\Gamma_n (m) \Gamma_n(n)}.
\end{equation}
In \eqref{eq.f}, $\mathbf{\Psi}_{n,m}(k, x)$ is an $n\times n$ matrix with
$(s,t)$th entry $[\Phi_{n,m}(k, x)]_{s,t}$ given by
\begin{equation*}
  \frac{[\Phi_{n,m}(k, x)]_{s,t}}{(n+m-s-t)!}=
  \begin{cases}
    \sum_{h=1}^{(n+m-s-t)+1} E_h(x), & t=k\\
    1 & t\ne k
  \end{cases}
\end{equation*}
where $E_h(\cdot)$ is the exponential integral function; and
\begin{equation}
\Gamma_n(m)=\prod_{i=1}^n \Gamma(m-i+1)
\end{equation}
with $\Gamma(\cdot)$ being the Gamma function.

We summarize the results in the following theorem.
\begin{theorem}
With $K$ mobile users and $M$ antennas at the base station, and a channel
model as given in \eqref{eq.sys}, coherence interval length $T$, training
energy $E$, and data transmission power $\Pd$, the following rates are
achievable: 1) rate given by \eqref{eq.ratemrc} with MRC receiver; 2) rate
given by \eqref{eq.ratezf} with ZF receiver; and 3) rate given by
\eqref{eq.ratemmse} with MMSE receiver.
\end{theorem}

As a numerical example, we depict in \figref{fig1} the rates that can be
achieved for a test system for optimized $\alpha$ as in \eqref{eq.optalpha},
as well as the unoptimized $\alpha=(T-K)/T$. It can be observed that the
achievable rate with optimized $\alpha$ is higher than the unoptimized one.
When SNR is large, the slops of the achievable with ZF and MMSE are the same
whatever energy slitting optimization is applied or not. We will analyze the
DoF in the next section.

\begin{figure}[!htp]
\centering
\includegraphics[width=\figwidth]{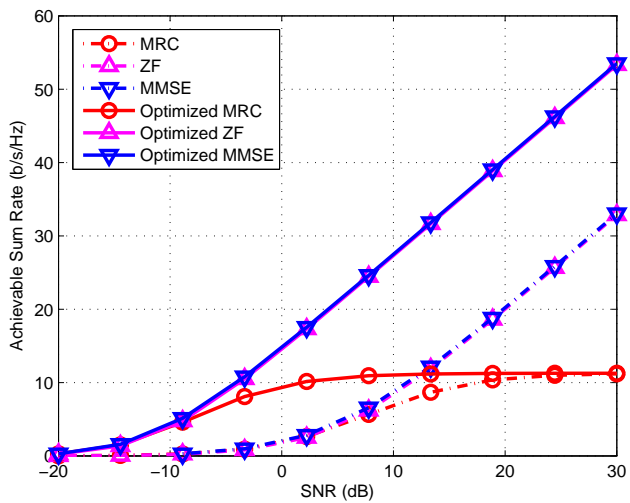}
\caption{Achievable sum rate of MRC, ZF and MMSE receivers, where $M=20$,
$K=4$, $T=196$. Solid lines indicate the sum rate with optimized $\alpha$ and
dashed lines illustrate the results with allocating the same power for both
training and data phases.}
\label{fig1}
\end{figure}

\subsection{Degrees of Freedom}

We define the DoF of the system as
\begin{equation}
  d(M, K, T)\bydef \sup \lim_{\rho\to\infty} \frac{R^{(\text{total})}(\rho)}{\log_2 (\rho)}
\end{equation}
where the supremum is taken over the totality of all reliable communication
schemes for the system, and $R^{(\text{total})}$ denotes the sum rate of the
$K$ users under the power constraint $\rho$. We may also speak of the
(achieved) degree of freedom of one user for a particular achievability
scheme, which is the achieved rate of the user normalized by $\log_2(\rho)$ in
the limit of $\rho\to \infty$. The DoF measures the multiplexing gain offered
by the system when compared to a reference point-to-point single-antenna
communication link, in the high SNR regime (see e.g., \cite{jafa11}).

\begin{theorem} \label{th.dof}
For an $(M, K, T)$ MIMO uplink system with $M$ receive antennas, $K$ users,
and coherence interval $T$, the total DoF of the system is
\begin{equation}
  d(M,K,T)=K^{\dag}\left(1-\frac {K^{\dag}}{T}\right).
\end{equation}
\end{theorem}
where $K^{\dag}\bydef \min(M, K, \lfloor T/2\rfloor)$. \QED

\begin{proof}
To prove the converse, we observe that if we allow the $K$ transmitters to
cooperate, then the system is a point-to-point MIMO system with $K$ transmit
antennas, $M$ receive antennas, and with no CSI at the receiver. The DoF of
this channel has been quantified in \cite{zhts02}, in the same form as in the
theorem. Without cooperation, the users can at most achieve a rate as high as
in the cooperation case.

To prove the achievability, we first look at the case $K^{\dag}<M$. In this
case, we note that if we allow only $K^{\dag}$ users to transmit, and let the
remaining users be silent, then using the achievability scheme describe in
\secref{channelesti}, each of the $K^{\dag}$ users can achieve a rate per user
using the zero-forcing receiver given as follows (cf.~\eqref{eq.ratezf})
\begin{equation}
\left (1 - \frac {K^{\dag}}T\right) \log_2\left(1 +
  \rhoe(M-K^{\dag}) \right).
\end{equation}
Note that the condition $K^{\dag}<M$ is needed. If we choose $E=K^{\dag}\rho$
and $\Pd=\rho$, then the effective SNR in \eqref{eq.eff} becomes
\begin{equation} \label{eq.highsnr}
\rhoe=\frac{\rho}{1+\frac{K\rho+1}{K\rho}}.
\end{equation}
It can be seen that as $\rho\to\infty$, $\log(\rhoe)/\log(\rho)\to 1$ and a
DoF per user of $(1-K^{\dag}/T)$ is achieved. The total achieved DoF is
therefore $K^{\dag}(1-K^{\dag}/T)$. Although better energy splitting is
possible, as in \secref{sec.split}, it will not improve the DoF.

When $K^{\dag}=M$, the case is more subtle. In this case the zero-forcing
receive is no longer sufficient. In fact, even the optimal linear processing,
which is the MMSE receiver \cite[eq.~(31)]{nglm13}, is not sufficient. The
insufficiency can be established by using the results in
\cite[Sec.IV.C]{gasc98} to show that as $\rho\to\infty$, the effective SNR at
the output of MMSE receiver has a limit distribution that is independent of
SNR. We skip the details here, since it is not the main concern in this paper.

Instead, we notice that the equivalent channel \eqref{eq.equiv} has SNR given
by \eqref{eq.highsnr}, which for $K\rho>1$ is greater than $\rho/3$. So, the
MIMO system can be viewed as a multiple access channel (MAC) with $K^{\dag}$
single-antenna transmitters, and one receiver with $M$ receive antennas.
Perfect CSI is known at the receiver, and the SNR between $\rho/3$ and $\rho$.
Using the MAC capacity region result \cite[Theorem~14.3.1]{coth91},
\cite[Sec.~10.2.1]{tse05}, it can be shown that a total DoF of $K^{\dag}$ can
be achieved over $T-K^{\dag}$ the time slots.
\end{proof}

\remark The DoF is the same as that of a point-to-point MIMO channel with $K$
transmit antennas and $M$ receive antennas without transmit- or receive-side
CSI \cite{zhts02}. This is a bit surprising because optimal signaling over
non-coherent MIMO channel generally requires cooperation among the transmit
antennas. It turns out that as far as DoF is concerned, transmit antenna
cooperation is not necessary. This is the new twist compared to the
point-to-point case.

\remark It can be seen from the achievability proof that for $M>K$, which is
generally applicable for massive MIMO systems, ZF at the base station is
sufficient for achieving the optimal DoF. However, MRC is not sufficient
because $\rho$ shows up both in the numerator and denominator of
\eqref{eq.ratemrc}. So as $\rho\to\infty$, the achieved rate is limited. This
is due to the interference among the users.

\remark For the case $K^{\dag}=M$, non-linear decoding such as successive
interference cancellation is needed.

\remark When $T$ is large, a per-user DoF close to 1 is achievable, as long as
$K\le M$.

\remark When $M$ is larger than $K^{\dag}$, increasing $M$ further has no
effect on the DoF. However, it is clear that more receive antennas is useful
because more energy is collected by additional antennas. We will discuss the
benefit of energy savings in the next section.

\subsection{Discussion} \label{discussion}

\subsubsection{Power Savings for Fixed Rate}\label{powersave}

As more antennas are added to the base station, more energy can be collected.
Therefore, it is possible that less energy is needed to be transmitted from
the mobile stations. When there is perfect CSI at the base station, it has
been shown in \cite{nglm13} that the transmission power can be reduced by a
factor $1/M$ to maintain the same rate, compared to a single-user
single-antenna system.

When there is no CSI at the receiver, however, it was observed in
\cite{nglm13} that the power savings factor is $1/\sqrt{M}$ instead of $1/M$.
In the following we do a slightly finer analysis of the effected power savings
when $M$ is large, assuming the training phase has been optimized as in
\secref{channelesti}.

Consider $M\gg K> 1$. Because the received power is linearly proportional to
$M$, the transmitted power can be smaller when $M$ is larger. When $M\gg 1$,
the system is operating in power-limited regime. It can be seen from
\eqref{eq.ratemrc} and \eqref{eq.ratezf} that when $\rho$ is small, MRC
performs better than ZF, which has been previously observed, e.g.,
\cite{nglm13}. On the other hand, in the low-SNR regime the difference between
them is a constant factor $(M-1)/(M-K)$ in the SNR term within the logarithmic
functions in \eqref{eq.ratemrc} and \eqref{eq.ratezf}. The difference becomes
negligible when $M$ is large. Using either result, and the effective SNR in
\eqref{eq.lowsnr}, we are able to obtain the following.
\begin{corollary}
If we fix the per-user rate at $R=(1-K/T)\log_2(1+\rho_0)$, then the required
power $\rho$ is
\begin{equation}\label{eq.saving}
  \rho=\sqrt{\frac{4\rho_0 (T-K)}{MT^2}} + o\left(\frac1{\sqrt M} \right)
\end{equation}
\end{corollary}
\begin{proof}
This can be proved by setting $\rho M=\rho_0$ in the rate expression for ZF.
Since the achievable rate with ZF processing is worse than MRC and MMSE when
SNR is very low, the result is still applied for MRC and MMSE processing.
\end{proof}

It is interesting to note that increasing $T$ has a similar effect as
increasing $M$ on the required transmission power, reducing the power by
$1/\sqrt{M}$ or $1/\sqrt{T}$. The reason is the if $T$ is increased, then the
energy that can be expended on training is increased, improving the quality of
channel estimation. On the other hand, for \eqref{eq.saving} to be applicable,
we need $M\gg K$.

\subsubsection{MMSE and Optimal Processing}

If MMSE processing is used at the base station, then the performance can be
improved compared to MRC and ZF. However, at low SNR, MRC is near optimal and
at high SNR, ZF is near optimal. So MMSE processing will not change the nature
of the results that we have obtained, although a slightly higher rate is
possible. Also, it is observed that the difference between ZF and MMSE is
negligible for a wide range of SNR as shown in Figure~\ref{fig1} and the
similar results can be found at \cite{nglm13}.


\subsubsection{Large Scale Fading} When large scale fading is considered, the
channel matrix becomes $\mathbf{G}=\mathbf{H}\mathbf{P}$. The matrix
$\mathbf{P}=\textrm{diag}\{\sqrt{p_1},\ldots,\sqrt{p_K}\}\in\mathbb{R}^{K\times
K}$ is diagonal where each entry models the path loss and shadow between the
base station and the $k$th user. The MMSE estimate of channel is given by
$\widehat{\mathbf{G}}=\frac{1}{\sqrt{E}}\mathbf{Y}_p(\frac{\mathbf{P}^{-1}}{E}+\mathbf{I})^{-1}$,
where the $k$th column of $\widehat{\mathbf{G}}$ is
\begin{equation}
\widehat{\mathbf{G}}_k=\frac{p^{\frac{3}{2}}_{k}E}{p_kE+1}\mathbf{h}_k+\frac{p_k\sqrt{E}}{p_kE+1}\mathbf{n}_k
\end{equation}
where $\mathbf{h}_k$ and $\mathbf{n}_k$ are the $k$th column of $\mathbf{H}$
and $\mathbf{N}$. With the definition of channel estimate error
$\widetilde{\mathbf{G}}=\mathbf{G}-\widehat{\mathbf{G}}$, we have the $k$th
column of $\widetilde{\mathbf{G}}$, i.e.,
\begin{equation}
\widetilde{\mathbf{G}}_k=\mathbf{G}_k-\widehat{\mathbf{G}}_k=\frac{p_k^{\frac{1}{2}}}{p_kE+1}\mathbf{h}_k-\frac{p_k\sqrt{E}}{p_kE+1}\mathbf{n}_k
\end{equation}
where $\mathbf{G}_k$ denotes the $k$th column of $\mathbf{G}$. Similar as in
\secref{mmsechannelesti}, we know that the elements of $\widehat{\mathbf{G}}$
and $\widetilde{\mathbf{G}}$ are independent complex Gaussian with zero mean.
Also, we can get the variances of each element of $\widehat{\mathbf{G}}_k$ and
$\widetilde{\mathbf{G}}_k$ are
\begin{equation}
\sigma^2_{\widehat{\mathbf{G}}_k}=\frac{p^2_kE}{p_kE+1}\quad\textrm{and}\quad\sigma^2_{\widetilde{\mathbf{G}}_k}=\frac{p_k}{p_kE+1}.
\end{equation}
Based on the definition of the equivalent channel in \secref{equichan}, we can
obtain the equivalent noise
\begin{equation}
 \sigma^2_{v}=\sum^K_{i=1}\frac{\rho_dp_i}{p_iE+1}+1
\end{equation}
and effective SNR of the $k$th user
\def\rhoek{\rho_{\text{eff},k}}
\begin{align}
 \rhoek\bydef&\frac {\Pd\sigma^2_{\widehat{\mathbf{G}}_k}}{\sigmav}=\frac{\rho_dp^2_kE}{(p_kE+1)(\sum^K_{i=1}\frac{\rho_dp_i}{p_iE+1}+1)}
 \nonumber\\
=&\frac{p^2_k\rho_d\rho_{\tau}T_{\tau}}{(p_k\rho_{\tau}T_{\tau}+1)(\sum^K_{i=1}\frac{\rho_dp_i}{p_i\rho_{\tau}T_{\tau}+1}+1)}.
\end{align}

Then, the corresponding achievable rates can be derived. Specifically, for MRC
the following ergodic rate of the $k$th user is achievable:
\def\rhoei{\rho_{\text{eff},i}}
\begin{align}\label{eq.ratemrck}
&R^\text{(MRC)}_k  \bydef\left (1 - \frac KT\right) \log_2\left(1 +
  \frac{\rhoek(M-1)}{\sum^M_{i=1,i\ne k}\rhoei+1} \right)
\nonumber\\
&=\left (1 - \frac KT\right) \log_2\bigg(1 +
\nonumber\\
&\quad\frac{p^2_k\rho_d\rho_{\tau}T_{\tau}(M-1)}{\rho_d(p_k\rho_{\tau}T_{\tau}+1)\sum^K_{i=1,i\ne k}p_k+p_k(\rho_d+\rho_{\tau}T_{\tau})+1}\bigg).
\end{align}
For ZF, assuming $M>K$, the following rate of the $k$th user is achievable:
\begin{align} \label{eq.ratezfk}
R^\text{(ZF)}_k \bydef &\left (1 - \frac KT\right) \log_2\left(1 +
  \rhoek(M-K) \right).
\end{align}

\remark The DoF is not changed when large scale fading is considered. Since we
can still choose $E=K^{\dag}\rho$ and $\Pd=\rho$, then the effective SNR in
\eqref{eq.ratezfk} becomes
\begin{equation} \label{eq.highsnrk}
\rhoek=\frac{\rho\frac{Kp^2_k}{Kp_k+1/\rho}}{\sum^K_{i=1}\frac{p_i}{Kp_i+1/\rho}+1}.
\end{equation}
Similar as in \thref{th.dof}, when $\rho\to\infty$,
$\log(\rhoek)/\log(\rho)\to 1$ and a DoF per user of $(1-K^{\dag}/T)$ is still
achieved.

\remark The convexity of $\rhoek$ in terms of $\alpha$ is not easy to prove,
but with chosen some special parameters, we can know $\rhoek$ is not concave.
But it can be proved as quasi-concave. Therefore, the achievable rate can be
improved with optimizing $\alpha$ but without global optimal guarantee.
However, for the MRC case, the sum achievable rate
$\sum^K_{k=1}R^\text{(MRC)}_k$ is concave in terms of $\alpha$ which can be
proved easily by \leref{le.1} in the \secref{sec.mrc}. Hence, the optimal
$\alpha$ can be also given in both case of average power constraint and peak
power constraint.

\remark In a practical system, the channel statistic information is provided
from downlink, and adaptive power control mechanism can be adopted for the
block fading channel. Since most of the effect of large scale fading can be
compensated \cite{gold05}, the power allocation for large scale fading within
one coherence interval is not the main issue for the single cell case.

\section{Joint Optimization of Energy Allocation and Training Duration}
\label{durapeak}

If the peak power, rather than the average power, is limited, then our DoF
result still holds because the achievability proof actually uses equal power
in the training and data transmission phases. The power savings discussion in
the previous subsection still applies, because the system is limited by the
total amount of energy available, and not how the energy is expended. In the
regime where the SNR is neither very high or very low, the peak power
constraint will affect the rate. Also, there is a peak power limit for
hardware implementation in practical. We provide a detailed analysis in this
section.

\subsection{Energy Allocation} We assume that the transmitters are subject to
both peak and average power constraints, where the peak power during the
transmission is assumed to be no more than $\rhomax$; i.e.,
\begin{equation}\label{eq.oricon1}
 0\le\rho_d, \rho_{\tau}\le \rhomax.
\end{equation}

\subsection{The Optimization Problem}

\def\objfunc#1{R^{#1}(\alpha,T_d)}
For an adopted receiver, $\mathcal{A}\in \{\text{MRC}, \text{ZF}\}$, our goal
is to maximize the uplink achievable rate subject to the peak and average
power constraints. Based on the model in \eqref{eq.equiv}, we will consider
two linear demodulation schemes: MRC and ZF receivers.

For MRC receiver, the received SNR for any of the $K$ users' symbols can be
obtained by substituting $\rhoe$ into $\rhoe(M-1)/(\rhoe(K-1)+1)$ (see
\cite[eq.~(39)]{nglm13}):
\begin{equation}\label{eq.mrcsinr}
\textsf{SNR}^{\text{(MRC)}}=
  \frac{T_{\tau}\rho_{\tau}\rho_d(M-1)}
    {T_{\tau}\rho_{\tau}\rho_d(K-1)+K\rho_d+T_{\tau}\rho_{\tau}+1}.
\end{equation}

For the ZF Receiver, the received SNR for any of the $K$ users' symbols can be
obtained by substituting $\rhoe$ into $\rhoe(M-K)$ (see
\cite[eq.~(42)]{nglm13}):
\begin{equation}\label{eq.zfrate}
\textsf{SNR}^{\text{(ZF)}}=
  \frac{T_{\tau}\rho_{\tau}\rho_d(M-K)}{K\rho_d+T_{\tau}\rho_{\tau}+1}.
\end{equation}

For either receiver, a lower bound on the sum rate achieved by the $K$ users
is given by
\begin{equation}\label{eq.R}
R^{(\mathcal{A})}(\alpha,T_d)=
  \frac{T_d}{T}K\log_2(1+\textsf{SNR}^{(\mathcal{A})})
\end{equation}
where $\mathcal{A}\in \{\text{MRC}, \text{ZF}\}$.

Our optimization problem can be formulated as follows:
\begin{align}\label{eq.oripro}
(\textbf{OP}) \qquad \maximize_{\alpha,T_d}\quad&\objfunc{(\mathcal A)} \\
\text{subject to}\quad & T_d+T_{\tau}=T \\
 & \rho T\alpha+\rhomax T_d  \le\rhomax T \label{eq.alcon1} \\
  -&\rho T\alpha-\rhomax T_d  \le-\rho T \label{eq.alcon2}  \\
 & 0\le \alpha  \le 1 \label{eq.alconalpha} \\
 & 0<T_d  \le T-K  \label{eq.tdconstr}
\end{align}
where $\objfunc{(\mathcal A)}$ is as given in \eqref{eq.R}; \eqref{eq.alcon1}
and \eqref{eq.alcon2} are from the peak power constraints in the training and
data phases, respectively; and the last constraint is from the requirement
that $T_\tau \ge K$.

\section{SNR Maximization when $T_d$ is Fixed} \label{sec.sola}

The feasible set of the problem (OP) is convex, but the convexity of the
objective function is not obvious. In this section, we consider the
optimization problem when $T_d$ is fixed. In this case, we will prove that
$R^{(\mathcal{A})}(\alpha,T_d)$ is concave in $\alpha$, and derive the
optimized $\alpha$. The result will be useful in the next section where
$\alpha$ and $T_d$ are jointly optimized.

For a fixed $T_d$, from the peak power constraints \eqref{eq.alcon1} and
\eqref{eq.alcon2}, we have
\begin{equation}\label{eq.solcon}
\frac{\rhomax T_{\tau}}{\rho T}+\left(1-\frac{\rhomax}{\rho}\right)\le\alpha\le\frac{\rhomax T_{\tau}}{\rho T}.
\end{equation}
Combined with \eqref{eq.alconalpha}, the overall constraints on $\alpha$ are
\begin{equation}\label{eq.solcon2}
\min\{0,\frac{\rhomax T_{\tau}}{\rho T}+\left(1-\frac{\rhomax}{\rho}\right)\}\le\alpha\le\max\{\frac{\rhomax T_{\tau}}{\rho T},1\}.
\end{equation}
In the remaining part of this section, we will first ignore the peak power
constraint, and derive the optimal $\alpha\in(0,1)$ for a given $T_d$. At the
end of this section, we will reconsider the effect of the peak power
constraint on the optimal $\alpha$.

\subsection{MRC Case without Peak Power Constraint} \label{sec.mrc}

Using \eqref{eq.optfra} we can rewrite \eqref{eq.mrcsinr} as
\begin{equation}\label{eq.newmrcsinr}
\textsf{SNR}^\text{(MRC)}(\alpha)=\frac{M-1}{K-1}
  \frac{\alpha(\alpha-1)}{\alpha^2-a_1\alpha-b_1}
\end{equation}
where
\begin{equation}\label{eq.defab}
a_1=1+\frac{T_d-K}{\rho T(K-1)},\quad b_1=\frac{\rho T K+T_d}{\rho^2T^2(K-1)}>0.
\end{equation}
It can be verified that $1-a_1-b_1\le0$.

\subsubsection{Behavior of the $\textsf{SNR}^\text{(MRC)}(\alpha)$ Function}

Define
\begin{equation}
g(\alpha):=\textsf{SNR}^\text{(MRC)}\cdot(K-1)/(M-1).
\end{equation}
And let $g_d(\alpha)=\alpha^2-a_1 \alpha -b_1$, which is the denominator of
$g(\alpha)$.

\begin{lemma}\label{le.1}
The function $g(\alpha)$ is concave in $\alpha$ over $(0,1)$ when
$1-a_1-b_1\le0$ and $b_1>0$.
\end{lemma}
\begin{proof}
See Appendix A.
\end{proof}
\leref{le.1} gives the convex conditions of the objective function. According
to \leref{le.1}, we know that there is a global maximal point for
\eqref{eq.newmrcsinr}. Taking the derivative of \eqref{eq.newmrcsinr} and
setting it as 0, we have
\begin{equation}\label{eq.equation}
(1-a_1)\alpha^2-2b_1\alpha+b_1=0.
\end{equation}

\remark \label{rem.noroot} It can be observed that when $1-a_1-b_1\le0$ and
$b_1>0$, $g_d(\alpha)$ is non-positive at both $\alpha=0$ and $\alpha=1$.
Since the leading coefficient of $g_d(\alpha)$ is positive, $g_d(\alpha)<0$
for $\alpha\in(0,1)$, and it has no root in $(0,1)$.

Based on \remref{rem.noroot}, we deduce that $g(\alpha)>0$ for $\alpha\in (0,
1)$. In addition, we have $g(0)=0$ and $g(1)=0$. Therefore, there is an
optimal $\alpha$ within $(0,1)$ rather than at boundaries.

\subsubsection{The Optimizing $\alpha$} We discuss the optimal $\alpha$ in
three cases, depending on $T_d$, as compared to $K$.
\begin{itemize}
\item If $T_d=K$, then $1-a_1=0$. Hence, we have $\alpha^*=1/2$, and
\begin{equation}
\textsf{SNR}^{\text{(MRC)}}(\frac{1}{2})=\frac{M-1}{K-1}\frac{1/4}{1/4+\frac{K(\rho T+1)}{\rho^2T^2(K-1)}}
\end{equation}

\item If $T_d<K$, then $1-a_1>0$. Since $b_1>1-a_1$, $b_1/(1-a_1)>1$. Between
the two roots of \eqref{eq.equation}, the one in between 0 and 1 is
\begin{equation}\label{eq.alphale}
\alpha^*=\frac{b_1-\sqrt{b_1(a_1+b_1-1)}}{1-a_1}.
\end{equation}

\item If $T_d>K$, then $1-a_1<0$. It can be deduced that in this case
$\alpha^*$ in \eqref{eq.alphale} is still between 0 and 1 and therefore is the
optimal $\alpha$.
\end{itemize}

Substituting \eqref{eq.defab} into \eqref{eq.alphale}, we have
\begin{equation}\label{eq.astar}
\alpha^*=\frac{\sqrt{(\rho TK+T_d)(\rho TT_d+T_d)}-(\rho TK+T_d)}{\rho T(T_d-K)}.
\end{equation}
We can simplify the expression for the optimal $\alpha$ at high and low SNR:
\begin{itemize}
\item At high SNR, the optimal $\alpha^*$ is
\begin{equation}
\alpha^*_{\rm H}\approx \frac{\sqrt{KT_d}-K}{T_d-K}=\frac{\sqrt{K}}{\sqrt{T_d}+\sqrt{K}}.
\end{equation}

\item Similarly, at low SNR, the optimal $\alpha^*$ is
\begin{align}
\alpha^*_{\rm L}\approx\frac{1}{2}.
\end{align}
As a result, $\textsf{SNR}^{\text{(MRC)}}(\alpha^*_{\rm
L})=(M-1)/(4T_d(K-1))$. If the SNR is low, the fraction between the training
and data is independent on the system parameters $M$, $K$, $\rho_d$,
$\rho_{\tau}$, $T_\tau$, and $T$.
\end{itemize}

\subsection{ZF Case without Peak Power Constraint} \label{sec.zf}

This optimization problem in the ZF case is similar to that in
Section.~\ref{sec.split}. Here, we only give the final optimization results.

Using \eqref{eq.optfra} we can rewrite \eqref{eq.zfrate} as
\begin{equation} \label{eq.zfsinr}
\textsf{SNR}^{\text{(ZF)}}(\alpha)=\frac{\rho T(M-K)\alpha(1-\alpha)}{(T_d-K)(\gamma+\alpha)}
\end{equation}

The optimized results of $\alpha^*$ are the same as in \eqref{eq.optalpha}.
The optimized $\textsf{SNR}^{\text{(ZF)}}$ is just given by $(M-K)\rhoe^*$. At
both high SNR and low SNR, the results are the same as in the MRC case.

\subsection{MRC and ZF with Peak Power Constraint}

So far we have ignored the peak power constraint. When the peak power is
considered, and $\alpha^*$ is not within the feasible set \eqref{eq.solcon2},
the optimal $\widetilde{\alpha}^*$ with the peak power constraint is the
$\alpha$ within the feasible set that is closest to the $\alpha^*$ we derived,
which is at one of the two boundaries of the feasible set, due to the
concavity of the objective function.

\section{Achievable Rate Maximization in General} \label{sec.solb}

In this section, $\alpha$ and $T_d$ are jointly optimized for maximizing the
achievable rate of uplink MU-MIMO system as illustrated in
\eqref{eq.oripro}--\eqref{eq.tdconstr} when both average and peak power
constraints are considered.

The feasible set with respective to $\alpha$ and $T_d$ is illustrated in
\figref{fig.feas}. It can be observed that the feasible region is in between
the following two lines
\begin{eqnarray}
T_d&=&-\rho T\alpha/\rhomax+T, \label{eq.slab1} \\ T_d&=&-\rho
T\alpha/\rhomax+\rho T/\rhomax \label{eq.slab2}
\end{eqnarray}
where $\alpha$ and $T_d$ satisfy \eqref{eq.alconalpha} and
\eqref{eq.tdconstr}.

\begin{figure}[htp]
\centering
\includegraphics[width=\iftoggle{twocolumn}{\linewidth}{.8\linewidth}]{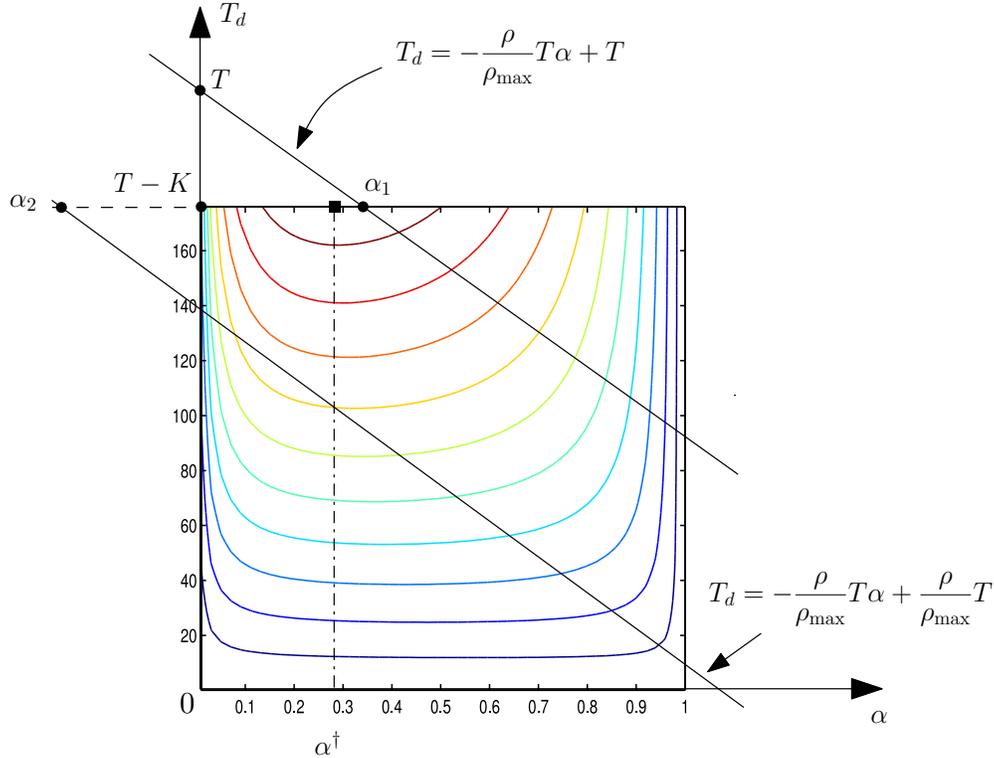}
\caption{Feasible region and the contour of the objective function in the MRC
case; $T=196$, $K=20$ and $M=50$.}
\label{fig.feas}
\end{figure}

We have the following lemma that is useful for describing the behavior of our
objective function $R^{(\mathcal A)}(\alpha, T_d)$ when $\alpha$ is fixed.

\begin{lemma} \label{lemma.2}
The function $f(x)=x\ln(1+a/(b+cx))$, when $a,b,c,x>0$, is concave and
monotonically increasing.
\end{lemma}
\begin{proof}
See Appendix B.
\end{proof}

In summary, the convexity of the objective function is known to have the
following two properties:
\begin{enumerate}
\item [(\textbf{P1})] From \leref{le.1}, for fixed $T_d$, $R^{(\mathcal{A})}$
is a concave function with respect to $\alpha$.

\item [(\textbf{P2})] From \leref{lemma.2}, for fixed $\alpha$,
$R^{(\mathcal{A})}$ is a concave function and monotonically increasing with
respect to $T_d$.
\end{enumerate}
Since the feasible set is convex, our optimization problem (OP) is a biconvex
problem that may include multiple local optimal solutions. However, after
studying the convexity of the objective function, there are only three
possible cases for the optimal solutions, as we discuss below.

In the remainder of this section, let $\adag$ denote the optimal $\alpha$ when
$T_d=T-K$, which is given by \secref{sec.mrc} and \secref{sec.zf} for MRC and
ZF processing.

\subsection{Case 1: $\rho_{\tau}$ is limited by $\rhomax$}

Define $\alpha_1:=\rhomax K/\rho T$, which is the root of $T-K=-\rho
T\alpha/\rhomax+T$ in $\alpha$ (see \figref{fig.feas}). In the case where
$\alpha_1<\adag$, because of the property P2 the optimal $(\alpha^*, T_d^*)$
must be on one of the two lines given by i) $T_d=-\rho T\alpha/\rhomax+T$,
$\alpha\in [\alpha_1, 1]$, and ii) $T_d=T-K$, $\alpha \in [0, \alpha_1]$.

On the line $T_d=T-K, \alpha\in[0,\alpha_1]$ the objective function is concave
and increasing with $\alpha$, thanks to property P1. Hence, we only need to
consider the line $T_d=-\rho T\alpha/\rhomax+T, \alpha\in[\alpha_1,1]$.

\begin{lemma}\label{le.quasi}
The objective function $\objfunc{(\text{MRC})}$ along the line $T_d=-\rho
T\alpha/\rhomax+T, \alpha\in[\alpha_1,1]$ is quasiconcave in $\alpha$.
\end{lemma}
\begin{proof}
Consider MRC processing. Substituting \eqref{eq.slab1} into
$\objfunc{(\text{MRC})}$, we have
\begin{equation}\label{eq.conmrcrate}
R^{\text{(MRC)}}(\alpha)=
  \frac{K}{T}\left(-\frac{\rho T}{\rhomax}\alpha+T\right)
  \log_2(1+\textsf{SNR}^{\text{(MRC)}}(\alpha))
\end{equation}
where
\begin{equation}\label{eq.sinrmrcq}
\textsf{SNR}^{\text{(MRC)}}(\alpha)=
  \frac{\alpha(\alpha-1)\rho^2T^2(M-1)}{a_2\alpha^2-b_2\alpha-c_2},
\end{equation}
and $a_2=\rho^2T^2(K-1)+\rho^2T^2/\rhomax$, $b_2=\rho^2T^2(K-1)+\rho T^2-\rho
TK-\rho T/\rhomax$ and $c_2=K\rho T+T$. Since $R^{\text{(MRC)}}(\alpha)>0$, in
order to prove the quasi-concavity of $R^{\text{(MRC)}}(\alpha)$, we need to
prove that the super-level set $\mathcal{S}_{\beta}=\{\alpha|0<\alpha<1,
R^{\text{(MRC)}}(\alpha)\ge\beta\}$ for each $\beta\in\mathbb{R}^{+}$ is
convex. Equivalently, if we define
\begin{equation}\label{eq.quasiproof}
\phi_{\beta}(\alpha)=
  \frac{\beta}{\frac{K}{T}(\frac{\rho T\alpha}{\rhomax}-T)}
    +\log_2(1+\textsf{SNR}^{\text{(MRC)}}(\alpha)).
\end{equation}
we only need to prove that
$\mathcal{S}_{\phi}=\{\alpha|0<\alpha<1,\phi_{\beta}(\alpha)\ge0\}$ is a
convex set.

It can be checked that the first part of $\phi_\beta(\alpha)$, namely
$\beta/[{\frac{K}{T}(\frac{\rho T\alpha}{\rhomax}-T)}]$, is concave for
$\alpha\in [0, 1]$. For the other part of $\phi_{\beta}(\alpha)$, from
\eqref{eq.sinrmrcq} we know that
\begin{equation}
a_2-b_2-c_2=\rho T(\frac{\rho}{\rhomax}-1)-T(1-\alpha\frac{\rho}{\rhomax})<0
\end{equation}
where $a_2,c_2>0$. Applying \leref{le.1}, we know
$\textsf{SNR}^{\text{(MRC)}}(\alpha)$ is concave. Hence,
$\log_2(1+\textsf{SNR}^{\text{(MRC)}}(\alpha))$ is also concave since function
$\log(1+x)$ is concave and non-decreasing \cite{bova11}. Therefore, its
super-level set $\mathcal{S}_{\phi}$ is convex. It follows that the
super-level set $\mathcal{S}_{\beta}$ of $R^{\text{(MRC)}}(\alpha)$ is convex
for each $\beta\ge 0$. The objective function is thus quasiconcave.
\end{proof}

Thanks to \leref{le.quasi}, we can find the optimal $\alpha$ by setting the
derivative of \eqref{eq.conmrcrate} with respect to $\alpha$ to 0. Efficient
one-dimensional searching algorithm such as Newton method or bisection
algorithm \cite{bova11}, can be adopted to find out the optimal $\alpha$.

\subsection{Case 2: $\rho_d$ is limited by $\rhomax$}

Define $\alpha_2:=1-\rhomax(T-K)/\rho T$, which is the root of $T-K=\rho
T\alpha/\rhomax+\rho T/\rhomax$ in $\alpha$. If $\alpha_2>\adag$, because of
the property P2 the optimal $(\alpha^*, T_d^*)$ must be on one of the two
lines given by i) $T_d=-\rho T\alpha/\rhomax+T, \alpha\in(\alpha_1,1)$,
$\alpha\in [\alpha_1, 1]$, and ii) $T_d=T-K$, $\alpha \in [\alpha_2,
\alpha_1]$. Along the line $T_d=T-K, \alpha\in(\alpha_1,1)$, the corresponding
function is decreasing in $\alpha$ because of the property P1. Also
considering P2, which implies that the optimal point in this case cannot
include $T_d<T-K$, we conclude that the point $(\alpha^*, T_d^*)=(\alpha_2,
T-K)$ is the global optimal solution of the problem.

\subsection{Case 3: Both $\rho_d$ and $\rho_{\tau}$ are not limited by
$\rhomax$}

If $\alpha_2<\adag<\alpha_1$, the optimal point is achieved at $(\alpha^*,
T_d^*)=(\adag, T-K)$, according to properties P1 and P2.

Summarizing what we have discussed so far, we have the following theorem.
\begin{theorem} \label{th.mrc}
For the MRC receiver, set $\adag=1/2$ if $T_d=K$ and otherwise set $\adag$
according to \eqref{eq.astar} when $T_d=T-K$. Set $\alpha_1=\rhomax K/\rho T$
and set $\alpha_2= 1-\rhomax (T-K)/\rho T$. The solution for the joint
optimization of training energy allocation $\alpha$ and the training duration
$T_\tau=T-T_d$ is given in three cases. Case~1) If $\alpha_1<\adag$, then
$\alpha^*$ is given by the maximizer of $R^{(\text{MRC})}(\alpha)$ in
\eqref{eq.conmrcrate}, and $T_d^*=-\rho T\alpha^* /\rhomax + T$; Case~2) If
$\alpha_2>\adag$ then $(\alpha^*, T_d^*)=(\alpha_2, T-K)$; Case~3) If
$\alpha_2<\adag<\alpha_1$, then $(\alpha^*, T_d^*)=(\adag, T-K)$.
\end{theorem}

We also have similar results regarding the optimal energy allocation factor
$\alpha$ and training period $T_\tau$ for the ZF case. The only difference is
that the achievable rate $R^{\text{(ZF)}}(\alpha)$ should be given by
substituting \eqref{eq.slab1} into $\objfunc{(\text{ZF})}$, which is
\begin{equation}\label{eq.conzfrate}
R^{\text{(ZF)}}(\alpha)=
  \frac{K}{T}\left(-\frac{\rho T}{\rhomax}\alpha+T\right)
  \log_2(1+\textsf{SNR}^{\text{(ZF)}}(\alpha))
\end{equation}
where
\begin{equation}\label{eq.sinrzfq}
\textsf{SNR}^{\text{(MRC)}}(\alpha)=
  \frac{\alpha(\alpha-1)\rho^2T^2(M-K)}{a_3\alpha^2-b_3\alpha-c_3},
\end{equation}
and $a_3=\rho^2T^2/\rhomax$, $b_3=\rho T^2-\rho TK-\rho T/\rhomax$ and
$c_3=K\rho T+T$. Comparing \eqref{eq.conmrcrate},~\eqref{eq.sinrmrcq} and
\eqref{eq.conzfrate},~\eqref{eq.sinrzfq}, we can obtain the results for ZF
receiver as follows.

\begin{theorem} \label{th.zf}
For the ZF receiver, set $\adag=1/2$ if $T_d=K$ and otherwise set $\adag$
according to \eqref{eq.optalpha} when $T_d=T-K$. Set $\alpha_1=\rhomax K/\rho
T$ and set $\alpha_2= 1-\rhomax (T-K)/\rho T$. The solution for the joint
optimization of training energy allocation $\alpha$ and the training duration
$T_\tau=T-T_d$ is given in three cases. Case~1) If $\alpha_1<\adag$, then
$\alpha^*$ is given by the maximizer of $R^{(\text{ZF})}(\alpha)$ in
\eqref{eq.conzfrate}, and $T_d^*=-\rho T\alpha^* /\rhomax + T$; Case~2) If
$\alpha_2>\adag$ then $(\alpha^*, T_d^*)=(\alpha_2, T-K)$; Case~3) If
$\alpha_2<\adag<\alpha_1$, then $(\alpha^*, T_d^*)=(\adag, T-K)$.
\end{theorem}

We also remark that our results are applicable for any $M>K$, including when
$M\gg K$, i.e., the massive MIMO system case.

\section{Discussion}\label{massivecase}

When $M$ increases, the transmit power of each user can be reduced
proportionally to $1/\sqrt{M}$ for large $M$ while maintaining a fixed rate as
discussed in \secref{powersave} and \cite{nglm13}. Here we discuss the
asymptotic achievable rates when $M\to\infty$.

\subsection{Optimized $\alpha$ if $T_d$ is fixed when $M\to\infty$} If the
energy over the training and data phases is allocated differently, we have the
following results after optimizing the $\alpha$ for large $M$.

\begin{theorem}
For both ZF and MRC, let $\rho_u:=\sqrt{M}\rho$ be fixed. Then, the maximum
achievable rate can be
\begin{equation}
R^{(\mathcal{A})}\rightarrow \frac{T_d}{T}K\log_2(1+\frac{\rho_u^2T^2}{4T_d}),\quad M\to\infty
\end{equation}
\end{theorem}

\begin{proof}
According to \eqref{eq.newmrcsinr} and \eqref{eq.zfsinr}, when $M\to\infty$,
we have
\begin{equation}
\textsf{SNR}^{(\mathcal{A})}(\alpha)=\frac{\alpha(1-\alpha)\rho_u^2T^2}{T_d},
\end{equation}
where the maximum received SNR can be obviously obtained when $\alpha=1/2$.
\end{proof}

Note, if the peak power constraints are considered, $\alpha$ needs to be
within the interval as shown in \eqref{eq.solcon2}. Otherwise, the optimal
solution is located at the boundary of \eqref{eq.solcon2}.

\remark If the power is allocated equally between the two phases, we have
$\alpha=T_{\tau}/T$ \cite{nglm13}, then the difference of achievable rate
between the optimized and the equally allocated power scheme is
\begin{align}
\Delta R^{(\mathcal{A})}(\alpha)&=\frac{T_d}{T}K(\log_2(1+\frac{\rho^2_uT^2}{4T_d})-\log_2(1+T_{\tau}\rho^2_u)),
\nonumber \\
&=\frac{T_d}{T}K\log_2(\frac{4T_d+\rho^2_uT^2}{4T_d+4T_d(T-T_d)\rho^2_u}),
\label{eq.gain}
\end{align}
where the numerator minus the denominator within the $\log_2(\cdot)$ is equal
to $\rho^2_u(T^2-4TT_d+T_d^2)=\rho^2_u(T-T_d)^2\ge0$. Therefore, it is clear
that the optimized achievable rate is always larger than the unoptimized one.
The gain in rate offered by optimizing the energy allocated for training is
given by \eqref{eq.gain}.

\subsection{Optimized $\alpha$ and $T_d$ when $M\to\infty$}

For both MRC and ZF, under the peak power constraints, the average transmit
power of each user is $\rho=\rho_u/\sqrt{M}$, where $\rho_u$ is fixed. Define
$\rho/\rho_{\max}=\xi$. Consequently, the corresponding
$\rho_{\max}=\rho_u/(\xi\sqrt{M})$. When $M\to\infty$, applying Theorems
\ref{th.mrc} and \ref{th.zf}, we have the following cases:
\begin{itemize}
\item \emph{Case 1: $\rho_{\tau}$ is limited by $\rhomax$}
\begin{equation}\label{eq.optialpha1}
R^{(\mathcal{A})}(\alpha)=K(-\xi\alpha+1)\log_2\left(1+\frac{\alpha(\alpha-1)\rho_u^2T}{\xi\alpha-1}\right)
\end{equation}
Taking the derivative of \eqref{eq.optialpha1} and setting it to zero, we can
obtain the optimal $\alpha$ with one dimension search algorithm \cite{bova11}.
Then, the duration $T^*_d$ can be obtained by \eqref{eq.slab1} directly with
substituting $\alpha^*$.

\item \emph{Case 2: $\rho_d$ is limited by $\rhomax$}
\begin{equation}\label{eq.optialpha}
R^{(\mathcal{A})}(\alpha^*)=K\xi(-\alpha^*+1)\log_2\left(1+\frac{\alpha^*(\alpha^*-1)\rho^2_uT}{T-K}\right)
\end{equation}
where $\alpha^*=1-(T-K)/(\xi T)$ and $T^*_d=T-K$.

\item \emph{Case 3: Neither $\rho_d$ nor $\rho_{\tau}$ is not limited by
$\rhomax$}
\begin{equation}
R^{(\mathcal{A})}(\alpha^*)=\frac{T-K}{T}K\log_2\left(1+\frac{\rho_u^2T^2}{4(T-K)}\right)
\end{equation}
where $\alpha^*=1/2$ and $T_d^*=T-K$.
\end{itemize}

\section{Numerical Results} \label{sec:simulations}

In this section, we compare the achievable rates between equal power
allocation scheme and our optimized one under average and peak power
constraints. In our simulations, we set $\rho_{\max}=1.2\rho$, $K=10$, and
$T=196$. We consider the following schemes: 1) MRC, which refers to the case
where MRC receiver is used and the same average power is used in both training
and data transmission phases \cite{nglm13}. 2) A-MRC, which refers to the case
where MRC receiver is used, the training duration is $K$, and there is only
average power constraint. 3) AP-MRC, where MRC receiver is used, and both the
training duration and training energy are optimized under both the average and
peak power constraints. We will also consider the ZF variants of the above
three cases, namely ZF, A-ZF, and AP-ZF. The energy efficiency is defined as
$\eta^{\mathcal A}:=\objfunc{\mathcal A}/\rho$.

\begin{figure}[thp]
\centering\includegraphics[width=\figwidth]{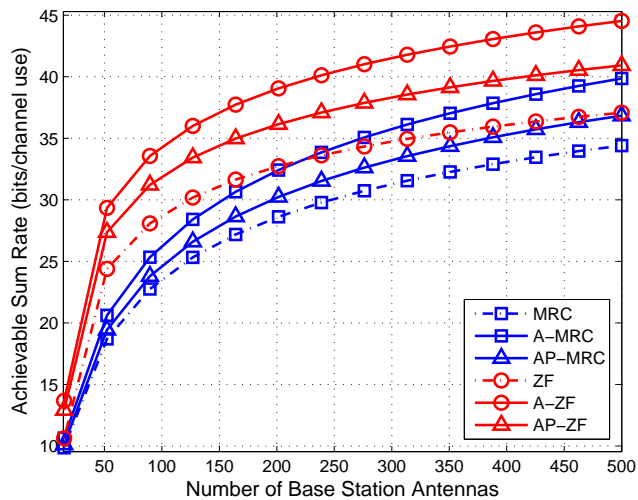}
\caption{Comparison between equal and optimized power allocations when the
number of base station antennas increases; $\rho_u=3$dB.}
\label{fig.image1}
\end{figure}

\begin{figure}[thp]
\centering \includegraphics[width=\figwidth]{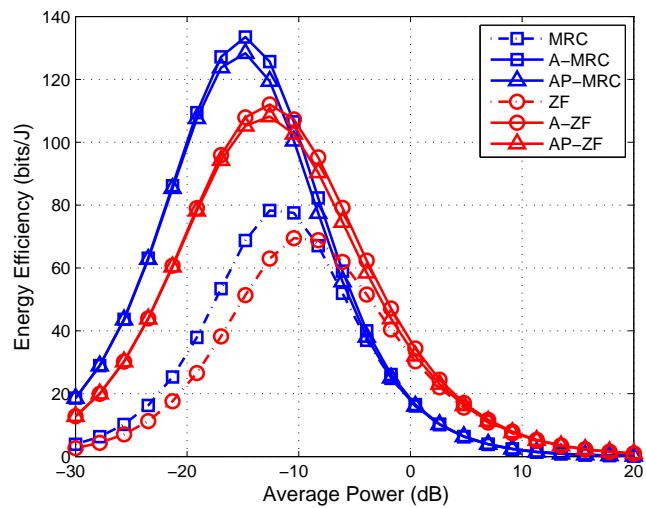}
\caption{Comparison of energy efficiency for different SNR $\rho$; where
$M=30$.} \label{fig.image2}
\end{figure}

In \figref{fig.image1}, we show the achieved rates of various schemes as the
number of antennas increases. It can be seen that the A-MRC (ZF) performs
better than the MRC (ZF) as well as the AP-MRC (ZF). In \figref{fig.image2},
the energy efficiency is shown as a function of $\rho$. It can be seen that
there is an optimal average transmitted power for maximum energy efficiency as
has been also observed before in \cite{nglm13}. It can also be seen that
optimized schemes, e.g., A-MRC (A-ZF) and AP-MRC (AP-ZF), show a significant
gain when SNR is low, since the power resource is scarce. Thus, the
optimization of power allocation and training duration plays much more
important role when $\rho$ is small than the case when $\rho$ is large. In
\figref{fig.image3}, we show the energy efficiency versus sum rate. In
particular, the optimized schemes achieve higher energy efficiencies. Also
from the simulations, we can see that ZF performs better than MRC at high SNR,
but worse when SNR is low.

Moreover, the impact of peak power constraint on achievable rates and energy
efficiencies for both MRC and ZF receivers can be observed through from
\figref{fig.image1} to \figref{fig.image3} clearly. They illustrate that when
peak power is limited at the training phase, the achievable rate with AP-MRC
and AP-ZF cannot be as high as the case with A-MRC and ZF. Although the
training period is increased, the time slot is still very precious when the
achievable rate needs to be maximized. In addition, we give an example about
energy efficiency versus the peak power limit in \figref{fig.image4}. It can
be observed that as the power limit increases, the energy efficiency with
AP-MRC and AP-ZF approaches the one with A-MRC and A-ZF, meaning that the
channel can be estimated accurately and more time slots are allocated to the
data phase.

\begin{figure}[htp]
\centering
\includegraphics[width=\figwidth]{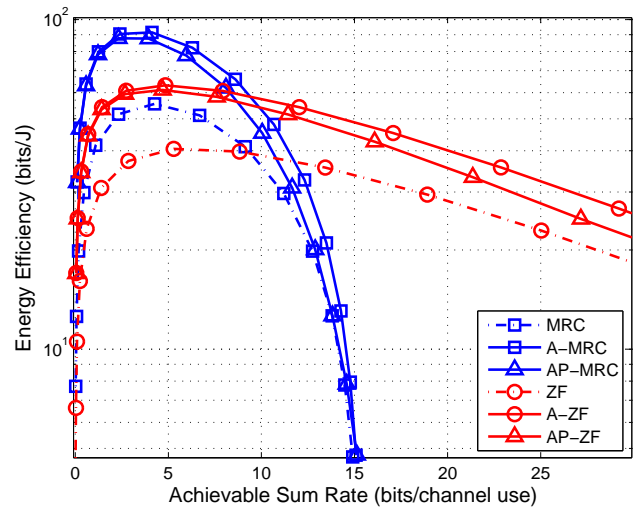}
\caption{Comparison of energy efficiency versus the achievable sum rate, where
$M=20$.}\medskip\label{fig.image3}
\end{figure}

\begin{figure}[htp]
\centering
\includegraphics[width=\figwidth]{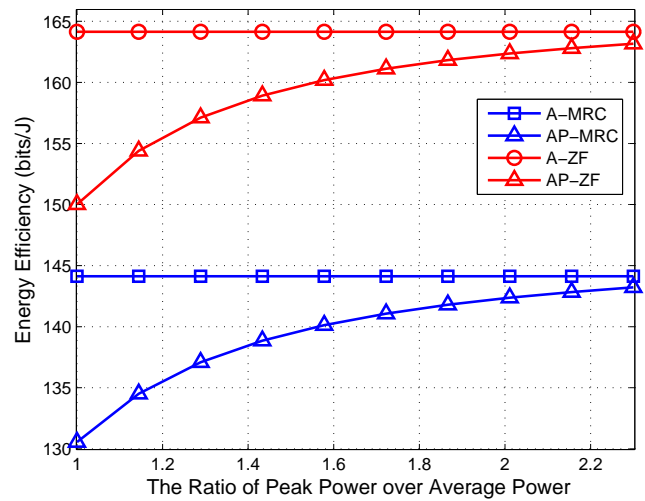}
\caption{Comparison of energy efficiency versus the ratio of peak power over
average power, where $M=50$ and $\rho=-10$dB.}\medskip\label{fig.image4}
\end{figure}

\section{Conclusions} \label{sec:conclusion}

In this paper, we considered an uplink multiuser cellular system where the
base station is equipped with multiple antennas. The channels were assumed to
be acquired by the base station through training symbols transmitted by the
mobile users. With estimated channels at the base station, we derived sum
rates that is achievable with MMSE channel estimation and MRC, ZF, and MMSE
detectors. Based on the derived rates, we were able to quantify the amount of
energy savings that are possible through the increase of either the number of
base station antennas, or the coherence interval length. We also quantified
the degrees of freedom that is possible in this scenarios, which is the same
as that of a point-to-point MIMO system. The achievability scheme when the
number of users is less than the number of base station antennas is linear:
zero-forcing is sufficient. Otherwise, nonlinear processing at the base
station is necessary to achieve the optimal total degrees of freedom.

For the case that both average and peak power constraints were considered, we
considered the problem of joint training energy and training duration
optimization for the MRC and ZF receivers so that the sum achievable rate is
maximized. We also performed a careful analysis of the convexity of the
problem and derived optimal solutions either in closed forms or in one case
through a one-dimensional search of a quasi-concave function. Our results were
illustrated and verified through numerical examples.

The effect of pilot contamination for multi-cell setup should also be the
importance influence on the achievable sum rate where the large scale fading
would be one of the main concerns, which will be considered as the future
work.

\section{Appendix}

\subsection{Proof of Lemma 1} Replacing $\alpha$ as $x$ in
\eqref{eq.newmrcsinr}, we need to verify that the second derivative of
\eqref{eq.newmrcsinr} with respective to $x$ is negative \cite{bova11}. The
first derivative of $g(x)$ is
\begin{equation}
g'(x)=\frac{(1-a)x^2-2bx+b}{(x^2-ax-b)^2}
\end{equation}
where $1-a-b<0$, $b>0$ and $x\in(0,1)$. Then, taking the second derivative of
$g(x)$, we have
\begin{equation}
g''(x)=\frac{2}{(x^2-ax-b)^3}\underbrace{\big((a-1)x^3+3bx^2-3bx+ab+b^2\big)}_{f(x)}
\end{equation}
From \remref{rem.noroot} and $b>0$, we know that $(x^2-ax-b)^3<0$. We need to
show $f(x)>0$.

Checking the boundary of $f(x)$, we know that
\begin{equation}
f(0)=ab+b^2=b(b+a)>b>0,
\end{equation}
\begin{equation}
f(1)=ab+b^2=a-1+ab+b^2=(a+b-1)(b+1)>0.
\end{equation}
Next, we need to consider the monotonicity of the function during the interval
$x\in(0,1)$. Taking the derivative of $f(x)$, we get
\begin{align}\label{eq.fxderiv}
f'(x)=&3(a-1)x^2+6bx-3b
\nonumber \\
&=3(a-1)(x^2+\frac{2b}{a-1}x-\frac{b}{a-1}),
\end{align}
which is a quadratic function.

When $a=1$, $f'(x)=6bx-3b=3b(2x-1)$. The function is decreasing until $x=1/2$
and increasing afterwards. Since
\begin{equation}
f(\frac{1}{2})=\frac{1}{4}b+b^2>0,
\end{equation}
it can be deduced that $f(x)>0$.

When $a\ne 1$, we know that $f'(1)=3(a+b-1)>0$, $f'(0)=-3b$, meaning that the
function $f(x)$ is decreasing first and increasing after the minimum point.

Here, we need to verify the minimum value of $f(x^*)$ is always greater than
0. According to \eqref{eq.fxderiv}, the minimum point given by the root of
$f'(x^*)=0$ is
\begin{equation}\label{eq.minpoint}
x^*=-\frac{b}{a-1}+\sqrt{\frac{b(a+b-1)}{(a-1)^2}},
\end{equation}
since $a+b>1$ and $b>0$. Substituting \eqref{eq.minpoint} into $f(x)$, we have
\begin{align}
f(x^*)= & [x((a-1)x^2+2bx-b)+bx^2-2bx+ab+b^2]\bigr|_{x=x^*} \nonumber \\
=
&\frac{b(a+b-1)}{a-1}\underbrace{(\frac{2b}{a-1}-2\sqrt{\frac{b(a+b-1)}{(a-1)^2}}+a)}_{h}.
\end{align}

For $a-1>0$,
\begin{align}
h=&\frac{2b}{a-1}-\frac{\sqrt{b(a+b-1)}}{a-1}+a
\nonumber \\
=&\frac{2}{a-1}\frac{b^2-b(a+b-1)}{b+\sqrt{b(a+b-1)}}+a \\
\mathop{>}\limits^{(a)}&a-\frac{2b}{b+\sqrt{b^2}}>0
\end{align}
where (a) is based on $a-1>0$. Therefore, $f(x)>0$.

For $a-1<0$,
\begin{align}
h&=\frac{2b}{a-1}+\frac{2\sqrt{b(a+b-1)}}{a-1}+a\\
&\mathop{<}^{(b)}\frac{2(1-a)}{a-1}+a<0,
\end{align}
where (b) is due to $b>1-a$. Hence, $f(x)>0$.

\subsection{Proof of Lemma 2} The derivative of $f(x)=x\ln(1+a/(b+cx))$, where
$a,b,c,x>0$, is
\begin{equation}
f'(x)=\ln(1+\frac{a}{cx+b})-\frac{acx}{(cx+a+b)(cx+b)}
\end{equation}
It is clear that $\lim_{x\to\infty}f'(x)=0$. If we can verify that the
function $f'(x)$ is monotonically decreasing, then $f'(x)$ is always positive.
Hence, we take the derivative of $f'(x)$, and obtain
\begin{equation} \label{eq.secderivtd}
f''(x)=-\frac{abc^2x+ac^2(a+b)x+2ac(a+b)b}{[(cx+b)(cx+a+b)]^2}<0,
\end{equation}
since $a,b,c,x>0$. This means that $f'(x)$ is decreasing. Therefore, $f'(x)$
is always positive, i.e., $f(x)$ is an increasing and concave function.

\bibliographystyle{IEEE}
\bibliography{refs}
\end{document}